\documentclass[10pt,a4paper]{article}
\usepackage{geometry}                
\usepackage{amsmath, amsthm} 
\usepackage{amsfonts,amssymb}
\usepackage{graphicx}
\usepackage{epstopdf}
\usepackage{color}
\usepackage{enumerate}
\usepackage{multirow} 

\newcommand{\PP}{\mathbb{P}}
\newcommand{\EE}{\mathbb{E}}

\newcommand{\RR}{\mathbb{R}}

\newcommand{\CC}{\mathbb{C}}
\newcommand{\NN}{\mathbb{N}}

\newcommand{\LL}{\mathcal{L}}

\newcommand{\F}{\mathcal{F}}

\newcommand{\tf}{\widetilde{f}}

\title{A Semi-Markovian Modeling of Limit Order Markets}

\author{{\bf Anatoliy Swishchuk and Nelson Vadori}\vspace{0.3cm}\\
Department of Mathematics and Statistics, University of Calgary,\\
2500 University Drive NW, Calgary, Alberta, Canada T2N 1N4\\
aswish@ucalgary.ca, nvadori@ucalgary.ca}

\date{March 11 2015}                                           

\newtheorem{thm}{Theorem}[part]

\newtheorem{lem}[thm]{Lemma}

\newtheorem{prop}[thm]{Proposition}
\newtheorem{rk}[thm]{Remark}

\begin{document}
\setlength{\parindent}{0mm}
\maketitle

{\bf Abstract.} R. Cont and A. de Larrard \cite{CL} introduced a tractable stochastic model for the dynamics of a limit order book, computing various quantities of interest such as the probability of a price increase or the diffusion limit of the price process. As suggested by empirical observations, we extend their framework to 1) arbitrary distributions for book events inter-arrival times (possibly non-exponential) and 2) both the nature of a new book event and its corresponding inter-arrival time depend on the nature of the previous book event. We do so by resorting to Markov renewal processes to model the dynamics of the bid and ask queues. We keep analytical tractability via explicit expressions for the Laplace transforms of various quantities of interest. We justify and illustrate our approach by calibrating our model to the five stocks Amazon, Apple, Google, Intel and Microsoft on June $21^{st}$ 2012. As in \cite{CL}, the bid-ask spread remains constant equal to one tick, only the bid and ask queues are modeled (they are independent from each other and get reinitialized after a price change), and all orders have the same size.\\
\vspace{2mm}

{\bf Key words.} limit order book, Markov renewal process, diffusion limit, duration analysis, Weibull, Gamma.

\vspace{2mm}

{\bf AMS subject classifications}. 60K15, 60K20, 90B22, 91B24, 91B70
 
\vspace{3mm}

\section{Introduction} 

Recently, interest in the modeling of limit order markets has increased. Some research has focused on optimal trading strategies in high-frequency environments: for example \cite{FP1} studies such optimal trading strategies in a context where the stock price follows a semi-Markov process, while market orders arrive in the limit order book via a point process correlated with the stock price itself. \cite{CJ} develops an optimal execution strategy for an investor seeking to execute a large order using both limit and market orders, under constraints on the volume of such orders. \cite{PSS} studies optimal execution strategies for the purchase of a large number of shares of a financial asset over a fixed interval of time.\\

On the other hand, another class of articles has aimed at modeling either the high-frequency dynamics of the stock price itself, or the various queues of outstanding limit orders appearing on the bid and the ask side of the limit order book, resulting in specific dynamics for the stock price. In \cite{FP}, a semi-Markov model for the stock price is introduced: the price increments are correlated and equal to arbitrary multiples of the tick size. The correlation between these price increments occurs via their sign only, and not their (absolute) size. In \cite{CST}, the whole limit order book is modeled (not only the ask and bid queues) via an integer-valued continuous-time Markov chain. Using a Laplace analysis, they compute various quantities of interest such as the probability that the mid-price increases, or the probability that an order is executed before the mid-price moves. A detailed section on parameter estimation is also presented. For a more thorough literature on limit order markets, we refer to the above cited articles and the references thereof.  \\

The starting point of the present manuscript is the article \cite{CL}, in which a stochastic model for the dynamics of the limit order book is presented. Only the bid and ask queues are modeled (they are independent from each other and get reinitialized after a price change), the bid-ask spread remains constant equal to one tick  and all orders have the same size. Their model is analytically tractable, and allows them to compute various quantities of interest such as the distribution of the duration between price changes, the distribution and autocorrelation of price changes, the probability of an upward
move in the price and the diffusion limit of the price process. Among the various assumptions made in this article, we seek to challenge two of them while preserving analytical tractability:
\begin{enumerate}[i)]
\item the inter-arrival times between book events (limit orders, market orders, order cancellations) are assumed to be independent and exponentially distributed.
\item the arrival of a new book event at the bid or the ask is independent from the previous events.
\end{enumerate}

Assumption i) is relatively common among the existing literature (\cite{G}, \cite{M}, \cite{CCM}, \cite{DW}, \cite{L}, \cite{SFGK}, \cite{CST}). Nevertheless, as it will be shown later, when calibrating the empirical distributions of the inter-arrival times to the Weibull and Gamma distributions (Amazon, Apple, Google, Intel and Microsoft on June 21st 2012), we find that the shape parameter is in all cases significantly different than 1 ($\sim 0.1$ to 0.3), which suggests that the exponential distribution is typically not rich enough to capture the behavior of these inter-arrival times.\\

Regarding Assumption ii), we split the book events into 2 different types: limit orders that increase the size of the corresponding bid or ask queue, and market orders/order cancellations that decrease the the size of the corresponding queue. Assimilating the former to the type "$+1$" and the latter to the type "$-1$", we find empirically that the probability to get an event of type "$\pm 1$" is not independent of the nature of the previous event. Indeed, we present below the estimated transition probabilities between book events at the ask and the bid for the stock Microsoft on June $21^{st}$ 2012. It is seen that the \textit{unconditional} probabilities $P(1)$ and $P(-1)$ to obtain respectively an event of type "$+1$" and "$-1$" are relatively close to $1/2$, as in \cite{CL}. Nevertheless, denoting $P(i,j)$ the \textit{conditional} probability to obtain an event of type $j$ given that the last event was of type $i$, we observe that $P(i,j)$ can significantly depend on the previous event $i$. For example, on the bid side, $P(1,1)=0.63$ whereas $P(-1,1) = 0.36$.\\

\begin{center}
\begin{tabular}{|c|c|c|}
\hline
 \textbf{Microsoft} & Bid&Ask\\\hline
$P(1,1)$ & 0.63 & 0.60 \\
$P(-1,1)$ & 0.36& 0.41\\\hline
$P(-1,-1)$ & 0.64& 0.59\\
$P(1,-1)$ &  0.37 & 0.40\\\hline
$P(1)$ & 0.49& 0.51\\
$P(-1)$ & 0.51 & 0.49\\\hline
\end{tabular}\\
$\left. \right.$\\
\scriptsize
\textit{Estimated probabilities for book event arrivals. June $21^{st}$ 2012.}\\
\end{center}
\vspace{5mm}

On another front, we will show that we can obtain diffusion limit results for the stock price without resorting to the strong symmetry assumptions of \cite{CL}. In particular, the assumption that price increments are i.i.d., which is contrary to empirical observations, as shown in \cite{FP} for example.\\

The paper is organized as follows: section 2 introduces our semi-Markovian modeling of the limit order book, section 3 presents the main probabilistic results obtained in the context of this semi-Markovian model (duration until the next price change, probability of price increase and characterization of the Markov renewal process driving the stock price process), section 4 deals with diffusion limit results for the stock price process, and section 5 presents some calibration results on real market data.\\

\section{A Semi-Markovian modeling of limit order markets}

Throughout this paper and to make the reading more convenient, we will use - when appropriate - the same notations as \cite{CL}, as it is the starting point of the present article. In this section we introduce our model, highlighting when necessary the mains differences with the model in \cite{CL}.\\

Let $s_t$, $s^a_t$, $s^b_t$ be respectively the mid, the ask and the bid price processes. Denoting $\delta$ the "tick size", these quantities are assumed to be linked by the following relations:
\begin{align*}
& s_t=\frac{1}{2}(s_t^a+s_t^b),
& s_t^a=s_t^b+\delta.
\end{align*}

We will also assume that $s_0^b$ is deterministic and positive. In this context, $s^a_0=s_0^b+\delta$ and $s_0=s_0^b+\frac{\delta}{2}$ are also deterministic and positive. As shown in \cite{CL}, the assumption that the bid-ask spread $s_t^a-s_t^b$ is constant and equal to one tick does not exactly match the empirical observations, but it is a reasonable assumption as \cite{CL} shows that - based on an analysis of the stocks Citigroup, General Electric, General Motors on June $26^{st}$ 2008 - more than 98\% of the observations have a bid-ask spread equal to 1 tick. This corresponds to a situation where the order book contains no empty levels (also called "gaps").\\

The price process $s_t$ is assumed to be piecewise constant: at random times $\{T_n\}_{n \geq 0}$ (we set $T_0:=0$), it jumps from its previous value $s_{T_n^-}$ to a new value $s_{T_n}=s_{T_n^-} \pm \delta$. By construction, the same holds for the ask and bid price processes $s^a_t$ and $s^b_t$. These random times $\{T_n\}$ correspond to the times at which either the bid or the ask queue get depleted, and therefore, the distribution of these times $\{T_n\}$ will be obtained as a consequence of the dynamics that we will choose to model the bid and ask queues. Let us denote $q_t^a$ and $q_t^b$ the non negative integer-valued processes  representing the respective sizes of the ask and bid queues at time $t$, namely the number of outstanding limit orders at each one of these queues. If the ask queue gets depleted before the bid queue at time $T_n$ - i.e. $q^a_{T_n}=0$ and $q^b_{T_n}>0$ - then the price goes up: $s_{T_n}=s_{T_n^-} + \delta$ and both queue values $(q^b_{T_n}, q_{T_n}^a)$ are immediately reinitialized to a new value drawn according to the distribution $f$, independently from all other random variables. In this context, if $n_b$, $n_a$ are positive integers, $f(n_b,n_a)$ represents the probability that, after a price increase, the new values of the bid and ask queues are respectively equal to $n_b$ and $n_a$. On the other hand, if the bid queue gets depleted before the ask queue at time $T_n$ - i.e. $q^a_{T_n}>0$ and $q^b_{T_n}=0$ - then the price goes down: $s_{T_n}=s_{T_n^-} - \delta$ and both queue values $(q^b_{T_n}, q_{T_n}^a)$ are immediately reinitialized to a new value drawn according to the distribution $ \tf$, independently from all other random variables. Following the previous discussion, one can remark that the processes $q^b_{t}$, $q_{t}^a$ will never effectively take the value 0, because whenever $q^b_{T_n}=0$ or $q^a_{T_n}=0$, we "replace" the pair $(q^b_{T_n}, q_{T_n}^a)$ by a random variable drawn from the distribution $f$ or $\tf$. The precise construction of the processes $(q^b_{t}, q_{t}^a)$ will be explained below.\\

Let $\tau_n:=T_n-T_{n-1}$ the "sojourn times" between two consecutive price changes, $N_t:= \sup\{n: T_n \leq t\}=\sup\{n: \tau_1+...+\tau_n \leq t\}$ the number of price changes up to time $t$, $X_n:=s_{T_n}-s_{T_{n-1}}$ the consecutive price increments (which can only take the values $\pm \delta$). With these notations we have:
\begin{align*}
& s_t=\sum_{k=1}^{N_t} X_k.\\
\end{align*}

Let us now present the chosen model for the dynamics of the bid and ask queues. As mentioned in introduction, we seek to extend the model \cite{CL} in the two following directions, as suggested by our calibration results:
\begin{enumerate}[i)]
\item inter-arrival times between book events (limit orders, market orders, order cancellations) are allowed to have an arbitrary distribution.
\item the arrival of a new book event at the bid or the ask and its corresponding inter-arrival time are allowed to depend on the nature of the previous event.
\end{enumerate}

In order to do so, we will use a Markov renewal structure for the joint process of book events and corresponding inter-arrival times occurring at the ask and bid sides. Formally, for the ask side, consider a family $\{R^{n,a}\}_{n \geq 0}$ of Markov renewal processes given by:
\begin{align*}
& R^{n,a}:=\{(V^{n,a}_k, T^{n,a}_k)\}_{k \geq 0}.
\end{align*}

For each $n$, $R^{n,a}$ will "drive" the dynamics of the ask queue on the interval $[T_n,T_{n+1})$ where the stock price remains constant. $\{V_k^{n,a}\}_{k \geq 0}$ and $\{T_k^{n,a}\}_{k \geq 0}$ represent respectively the consecutive book events and the consecutive inter-arrival times between these book events at the ask side on the interval $[T_n,T_{n+1})$. At time $T_{n+1}$ where one of the bid or ask queues gets depleted, the stock price changes and the model will be reinitialized with an independent copy $R^{n+1,a}$ of $R^{n,a}$: it will therefore be assumed that the processes $\{R^{n,a}\}_{n \geq 0}$ are independent copies of the same Markov renewal process of kernel $Q^a$, namely for each $n$:
\begin{align*}
& \PP[V_{k+1}^{n,a}=j, T_{k+1}^{n,a}  \leq t| T_{p}^{n,a}, V_{p}^{n,a}: p \leq k]=Q^a(V_{k}^{n,a},j,t), \hspace{5mm} j \in \{-1,1\}\\
& \PP[V_{0}^{n,a}=j]=v^a_0(j), \hspace{5mm} j \in \{-1,1\}\\
& \PP[T_{0}^{n,a}=0]=1.
\end{align*}

We recall that as mentioned earlier, we consider two types of book events $V_{k}^{n,a}$: events of type $+1$ which increase the ask queue by 1 (limit orders), and events of type $-1$ which decrease the ask queue by 1 (market orders and order cancellations). In particular, the latter assumptions constitute a generalization of \cite{CL} in the sense that for each $n$:
\begin{itemize}
 \item $V_{k+1}^{n,a}$ depends on the previous queue change $V_{k}^{n,a}$: $\{V_k^{n,a}\}_{k \geq 0}$ is a Markov chain.
 \item the inter-arrival times $\{T_k^{n,a}\}_{k \geq 0}$ between book events can have arbitrary distributions. Further, they are not strictly independent anymore but they are independent conditionally on the Markov chain $\{V_k^{n,a}\}_{k \geq 0}$.
\end{itemize}

We use the same notations to model the bid queue - but with indexes $^a$ replaced by $^b$ - and we assume that the processes involved at the bid and at the ask are independent.\\

In \cite{CL}, the kernel $Q^a$ is given by (the kernel $Q^b$ has a similar expression with indexes $^a$ replaced by $^b$):
\begin{align*}
& Q^a(i,1,t)=\frac{\lambda^a}{\lambda^a+\theta^a+\mu^a} (1-e^{-(\lambda^a+\theta^a+\mu^a)t}), \hspace{5mm} i \in \{-1,1\}\\
& Q^a(i,-1,t)=\frac{\theta^a+\mu^a}{\lambda^a+\theta^a+\mu^a} (1-e^{-(\lambda^a+\theta^a+\mu^a)t}) , \hspace{5mm} i \in \{-1,1\}.
\end{align*}

Given these chosen dynamics to model to ask and bid queues between two consecutive price changes, we now specify formally the "state process":
\begin{align*}
& \widetilde{L}_t:=(s_t^b,q_t^b, q_t^a)
\end{align*}

which will keep track of the state of the limit order book at time $t$ (stock price and sizes of the bid and ask queues). In the context of \cite{CL}, this process $\widetilde{L}_t$ was proved to be Markovian. Here, we will need to "add" to this process the process $(V_t^b, V_t^a)$ keeping track of the nature of the last book event at the bid and the ask to make it Markovian: in this sense we can view it as being semi-Markovian. The process:
\begin{align*}
& L_t:=(s_t^b,q_t^b, q_t^a, V_t^b, V_t^a)
\end{align*}
constructed below will be proved to be Markovian.\\

The process $L$ is piecewise constant and changes value whenever a book event occurs at the bid or at the ask. We will construct both the process $L$ and the sequence of times $\{T_n\}_{n \geq 0}$ recursively on $n \geq 0$. The recursive construction starts from $n=0$ where we have $T_0=0$, $s_0^b>0$ deterministic, and $(q_0^b, q_0^a, V_0^b, V_0^a)$ is a random variable with distribution $f_0 \times v^b_0 \times v^a_0$, where $f_0$ is a distribution on $\NN^* \times \NN^*$, and both $v^b_0$ and $v^a_0$ are distributions on the two-point space $\{-1,1\}$, that is $v^b_0(1)=\PP[V_0^b=1]$ is given and $v^b_0(-1)=1-v^b_0(1)$ (and similarly for the ask). We will need to introduce the following processes for the bid side (for the ask side, they are defined similarly): 
\begin{align*}
& \bar{T}_k^{n,b}:=\sum_{p=0}^k T_p^{n,b},
& N^{n,b}_t:= \sup\{k: T_n+\bar{T}_k^{n,b} \leq t\}.
\end{align*}

With these notations, the book events corresponding to the interval $[T_n,T_{n+1})$ occur at times $T_n+\bar{T}_k^{n,b}$ ($k \geq 0$) until one of the queues gets depleted, and $N^{n,b}_t$ counts the number of book events on the interval $[T_n,t]$, for $t \in [T_n,T_{n+1})$. \\

The joint construction of $L$ and of the sequence of times $\{T_n\}_{n \geq 0}$ is done recursively on $n \geq 0$. The following describes the step $n$ of the recursive construction:
\begin{itemize}
\item For each $T \in \{T_n+\bar{T}_k^{n,b}\}_{k \geq 1}$, the book event $v^b_{n,T}:=V_{N^{n,b}_{T}}^{n,b}$ occurs at time $T$ at the bid side. If $q_{T^-}^b+v^b_{n,T}>0$, there is no price change at time $T$ and we have:
\begin{align*}
& (s^b_T,q_T^b,q_T^a,V_T^b, V_T^a)=(s^b_{T^-},q_{T^-}^b+v^b_{n,T},q_{T^-}^a,v^b_{n,T},V_{T^-}^a).
\end{align*}

If on the other hand $q_{T^-}^b+v^b_{n,T}=0$, there is a price change at time $T$ and the model gets reinitialized:
\begin{align*}
& (s^b_T,q_T^b,q_T^a,V_T^b, V_T^a)=(s^b_{T^-}-\delta , \tilde{x}^b_{n},\tilde{x}^a_{n}, v^b_{0,n}, v^a_{0,n}),
\end{align*}
where $\{(\tilde{x}_k^b,\tilde{x}_k^a)\}_{k \geq 0}$ are i.i.d. random variables, independent from all other random variables, with joint distribution $\tf$ on $\NN^* \times \NN^*$, and $\{v^b_{0,k}, v^a_{0,k}\}_{k \geq 0}$ are i.i.d. random variables, independent from all other random variables, with joint distribution $v^b_0 \times v^a_0$ on the space $\{-1,1\} \times \{-1,1\}$. We then set $T_{n+1}=T$ and move from the step $n$ of the recursion to the step $n+1$.

\item For each $T \in \{T_n+\bar{T}_k^{n,a}\}_{k \geq 1}$, the book event $v^a_{n,T}:=V_{N^{n,a}_{T}}^{n,a}$ occurs at time $T$ at the ask side. If $q_{T^-}^a+v^a_{n,T}>0$, there is no price change at time $T$ and we have:
\begin{align*}
& (s^b_T,q_T^b,q_T^a,V_T^b, V_T^a)=(s^b_{T^-},q_{T^-}^b,q_{T^-}^a+v^a_{n,T},V_{T^-}^b,v^a_{n,T}).
\end{align*}

If on the other hand $q_{T^-}^a+v^a_{n,T}=0$, there is a price change at time $T$ and the model gets reinitialized:
\begin{align*}
& (s^b_T,q_T^b,q_T^a,V_T^b, V_T^a)=(s^b_{T^-}+\delta , x^b_{n},x^a_{n}, v^b_{0,n}, v^a_{0,n}),
\end{align*}
where $\{(x_k^b,x_k^a)\}_{k \geq 0}$ are i.i.d. random variables, independent from all other random variables, with joint distribution $f$ on $\NN^* \times \NN^*$, and $\{v^b_{0,k}, v^a_{0,k}\}_{k \geq 0}$ are the i.i.d. random variables defined above. We then set $T_{n+1}=T$ and move from the step $n$ of the recursion to the step $n+1$.
\end{itemize}

It results from the above construction and the Markov renewal structure of the processes $\{R^{n,a}\}_{n \geq 0}$, $\{R^{n,b}\}_{n \geq 0}$ that the process $L_t$ is Markovian.\\

Since the processes $\{R^{n,a}\}_{n \geq 0}$ are independent copies of the same Markov renewal process of kernel $Q^a$, we will drop the index $n$ when appropriate in order to make the notations lighter. Following this remark, we will introduce the following notations for the ask, for $i,j \in \{-1,1\}$ (for the bid, they are defined similarly):
\begin{align*}
& P^a(i,j):=\PP[V_{k+1}^a=j | V_{k}^a=i],\\
& F^a(i,t):=\PP[T_{k+1}^a \leq t |�V_{k}^a=i],\\
& H^a(i,j,t):=\PP[T_{k+1}^a \leq t |�V_{k}^a=i,V_{k+1}^a=j],\\
& h^a(i,j):=\int_0^\infty t H^a(i,j,dt),\\
& h_1^a:=h^a(1,1)+h^a(-1,-1), \hspace{5mm}h_2^a:=h^a(-1,1)+h^a(1,-1),\\
& m^a(s,i,j):=\int_0^\infty e^{-st} Q^a(i,j,dt), \hspace{3mm} s \in \CC,\\
& M^a(s,i):=m^a(s,i,-1)+m^a(s,i,1)=\int_0^\infty e^{-st} F^a(i,dt), \hspace{3mm} s \in \CC.
\end{align*}

Throughout this paper, we will use the following mild technical assumptions:\\

{\bf (A1)} $0<P^a(i,j)<1$, $0<P^b(i,j)<1$, \hspace{3mm} $i,j \in \{-1,1\}$.\\

{\bf (A2)} $F^a(i,0)<1$, $F^b(i,0)<1$, \hspace{3mm} $i \in \{-1,1\}$.\\

{\bf (A3)} $\int_0^\infty t^2 H^a(i,j,dt)<\infty$, $\int_0^\infty t^2 H^b(i,j,dt)<\infty$, \hspace{3mm} $i,j \in \{-1,1\}$.\\

Some brief comments on these assumptions: \textbf{(A1)} implies that each state $\pm 1$ is accessible from each state. \textbf{(A2)} means that each inter-arrival time between book events has a positive probability to be non zero, and \textbf{(A3)} constitutes a second moment integrability assumption on the cumulative distribution functions $H^a$ and $H^b$.

\vspace{3mm}

\section{Main Probabilistic Results}

Throughout this section and as mentioned earlier, since the processes $\{R^{n,a}\}_{n \geq 0}$ are independent copies of the same Markov renewal process of kernel $Q^a$, we will drop the index $n$ when appropriate in order to make the notations lighter on the random variables 
$T_k^{n,a}$, $\bar{T}_k^{n,a}$, $V_k^{n,a}$ (and similarly for the bid side).\\

\subsection{Duration until the next price change}

Given an initial configuration of the bid and ask queues $(q_0^b,q_0^a)=(n_b,n_a)$ ($n_b,n_a$ integers), we denote $\sigma_b$ the first time at which the bid queue is depleted: 
\begin{align*}
& \sigma_b=\bar{T}_{k^*}^{b},
&k^*:=\inf\{k: n_b+\sum_{m=1}^k V^b_m=0\}.
\end{align*}

Similarly we define $\sigma_a$ the first time at which the ask queue is depleted. The duration until the next price move is thus:
\begin{align*}
& \tau:= \sigma_a \wedge \sigma_b.\\
\end{align*}
In order to have a realistic model in which the queues always get depleted at some point, i.e. $\PP[\sigma_a<\infty]=\PP[\sigma_b<\infty]=1$, we impose the conditions:
\begin{align*}
& P^a(1,1) \leq P^a(-1,-1),
& P^b(1,1) \leq P^b(-1,-1).\\
\end{align*}

These conditions correspond to the condition $\lambda \leq \theta + \mu$ in \cite{CL}, and the proof of the proposition below shows that they are respectively equivalent to $\PP[\sigma_a<\infty]=1$ and $\PP[\sigma_b<\infty]=1$. Indeed, as $s \to 0$ ($s>0$), the Laplace transform $\LL^a(s):=\EE[e^{-s \sigma_a}]$ of $\sigma_a$ tends to $\PP[\sigma_a<\infty]$. The proposition below shows that if $P^a(1,1) > P^a(-1,-1)$, this quantity is strictly less than 1, and if $P^a(1,1) \leq P^a(-1,-1)$, this quantity is equal to 1. We have the following result which generalizes the Proposition 1 in \cite{CL} (see also Remark \ref{lob3} below):

\begin{prop}
\label{lob1}
The conditional law of $\sigma_a$ given $q_0^a=n \geq 1$ has a regularly varying tail with:
\begin{itemize}
\item tail exponent 1 if $P^a(1,1) < P^a(-1,-1)$.
\item tail exponent 1/2 if $P^a(1,1) = P^a(-1,-1)$.
\end{itemize}
More precisely, we get: if $P^a(1,1) = P^a(-1,-1)=p_a$:
\begin{align*}
& \PP[\sigma_a>t | q_0^a=n ] \stackrel{t \to \infty}{\sim} \frac{\alpha^a(n)}{\sqrt{t}}
\end{align*}

with:
\begin{align*}
& \alpha^a(n):= \frac{1}{p_a \sqrt{\pi}}(n+\frac{2p_a-1}{p_a-1} v^a_0(1)) \sqrt{p_a(1-p_a)}\sqrt{p_a h_1^a+(1-p_a)h_2^a}.\\
\end{align*}

If $P^a(1,1) < P^a(-1,-1)$, we get:
\begin{align*}
& \PP[\sigma_a>t | q_0^a=n ] \stackrel{t \to \infty}{\sim}  \frac{\beta^a(n)}{t} 
\end{align*}

with:
\begin{align*}
& \beta^a(n):=v_0^a(1)u^a_1+v_0^a(-1)u^a_2+(n-1)u^a_3,\\
& u^a_1:=h^a(1,-1)+\frac{P^a(1,1)}{1-P^a(1,1)}(u^a_3+h^a(1,1))\\
& u^a_2:=-h^a(1,1)+\frac{1-P^a(-1,-1)}{1-P^a(1,1)}(u^a_3+h^a(1,1))+P^a(-1,-1)h_1^a+(1-P^a(-1,-1)) h_2^a,\\
& u^a_3:=h^a(1,1)+\frac{1-P^a(1,1)}{P^a(-1,-1)-P^a(1,1)} \left( P^a(-1,-1) h_1^a+(1-P^a(-1,-1)) h_2^a\right).\\
\end{align*}

Similar expressions are obtained for $\PP[\sigma_b>t | q_0^b=n ]$, with indexes $^a$ replaced by $^b$. \\

\end{prop}

\begin{rk}
\label{lob3}
We retrieve the results of \cite{CL}: if $P^a(1,1)=P^a(-1,-1)$, then within the context/notations of \cite{CL} we get $p_a=1/2$ and:
\begin{align*}
& h^a(i,j)=\int_0^\infty 2t \lambda e^{-2 \lambda t}dt=\frac{1}{2 \lambda},
\end{align*}
and so $\alpha^a(n)=\frac{n}{\sqrt{\pi \lambda}}$. For the case $P^a(1,1)<P^a(-1,-1)$ ($\lambda<\theta+\mu$ with their notations), we find:
\begin{align*}
& \beta^a(n)=\frac{n}{\theta+\mu-\lambda},
\end{align*}
which is different from the result of \cite{CL} that is $\beta^a(n)=\frac{n(\theta+\mu+\lambda)}{2\lambda(\theta+\mu-\lambda)}.$. We believe that they made a small mistake in their Taylor expansion on page 10: in the case $\lambda<\theta+\mu$, they should find:
\begin{align*}
& \LL(s,x) \stackrel{s \to 0}{\sim} 1-\frac{sx}{\theta+\mu-\lambda}.
\end{align*}
\end{rk}

\begin{proof}  
Let $s>0$ and denote $\LL(s,n,i):=\EE[e^{-s \sigma_a} |q_0^a=n,V_0^a=i]$. We have:
\begin{align*}
& \sigma_a=\sum_{m=1}^{k^*} T^a_m,
&k^*:=\inf\{k: n+\sum_{m=1}^k V^a_m=0\}.
\end{align*}

Therefore:
\begin{align*}
& \LL(s,n,i)=\EE[e^{-s T_{1}^a} \EE[e^{-s (\sigma_a-T_{1}^a)} | q_0^a=n,V_0^a=i,V_1^a,T_{1}^a ] |q_0^a=n,V_0^a=i]\\
&=\EE[e^{-s T_{1}^a} \underbrace{\EE[e^{-s (\sigma_a-T_{1}^a)} | q_{T_{1}^a}^a=n+V_1^a,V_0^a=i,V_1^a,T_{1}^a ]}_{\LL(s,n+V_1^a,V_1^a)} |q_0^a=n,V_0^a=i]\\
&= \EE[e^{-s T_{1}^a}  \LL(s,n+V_1^a,V_1^a) |q_0^a=n,V_0^a=i]\\
&= \int_0^\infty e^{-st} \LL(s,n+1,1) Q^a(i,1,dt)+\int_0^\infty e^{-st} \LL(s,n-1,-1) Q^a(i,-1,dt)\\
&= m^a(s,i,1) \LL(s,n+1,1)+m^a(s,i,-1) \LL(s,n-1,-1)
\end{align*}

Denote for sake of clarity $a_n:=\LL(s,n,1)$, $b_n:=\LL(s,n,-1)$. These sequences therefore solve the system of coupled recurrence equations:
\begin{align*}
& a_{n+1}=m^a(s,1,1) a_{n+2}+m^a(s,1,-1) b_{n}, \hspace{5mm} n \geq 0\\
& b_{n+1}=m^a(s,-1,1) a_{n+2}+m^a(s,-1,-1) b_{n}\\
& a_0=b_0=1.
\end{align*}

Simple algebra (computing $a_{n+1}-m^a(s,-1,-1)a_n$ on the on hand and $m^a(s,1,1)b_{n+1}-b_n$ on the other hand) gives us that both $a_n$ and $b_n$ solve the same following recurrence equation (but for different initial conditions):
\begin{align*}
& m^a(s,1,1) u_{n+2}-(1+\Delta^a(s))u_{n+1}+m^a(s,-1,-1) u_{n}, \hspace{5mm} n \geq 1\\
\end{align*}

with:
\begin{align*}
& \Delta^a(s):=m^a(s,1,1)m^a(s,-1,-1)-m^a(s,-1,1)m^a(s,1,-1).\\
\end{align*}

The parameter $\Delta^a(s)$ can be seen as a coupling coefficient and is equal to 0 when the random variable $(V^a_{k},T^a_k)$ doesn't depend on the previous state $V^a_{k-1}$, for example in the context of \cite{CL}.\\

If we denote $R(X)$ the characteristic polynomial associated to the previous recurrence equation $R(X):=m^a(s,1,1) X^2-(1+\Delta^a(s))X+m^a(s,-1,-1)$, then simple algebra gives us:
\begin{align*}
& R(1)=\underbrace{(M^a(s,1)-1)}_{< 0} \underbrace{(1-m^a(s,-1,-1))}_{>0}+\underbrace{m^a(s,1,-1)}_{>0} \underbrace{(M^a(s,-1)-1)}_{< 0} <0\\
\end{align*}

Note that $M^a(s,i)<1$ for $s>0$ because $F^a(i,0)<1$. Since $m^a(s,1,1)>0$, this implies that $R$ has only one root $<1$ (and an other root $>1$):
\begin{align*}
& \lambda^a(s):=\frac{1+\Delta^a(s)-\sqrt{(1+\Delta^a(s))^2-4m^a(s,1,1)m^a(s,-1,-1)}}{2m^a(s,1,1)}.
\end{align*}

Because we have $a_n,b_n \leq 1$ for $s>0$, then we must have for $n \geq 1$:
\begin{align*}
& a_n=a_1 \lambda^a(s)^{n-1}
& b_n=b_1 \lambda^a(s)^{n-1}\\
\end{align*}

The recurrence equations on $a_n$, $b_n$ give us:
\begin{align*}
& a_1=\frac{m^a(s,1,-1)}{1-\lambda^a(s)m^a(s,1,1) }
& b_1=\frac{m^a(s,-1,1)a_1+\Delta^a(s)}{m^a(s,1,1)}
\end{align*}

Finally, letting $\LL(s,n):=\EE[e^{-s \sigma_a} |q_0^a=n]$, we obtain:
\begin{align*}
& \LL(s,n)=\sum_{i} \LL(s,n,i) v^a_0(i)=a_n v^a_0(1)+b_n v^a_0(-1).
\end{align*}

The behavior of $\PP[\sigma_a>t | q_0^a=n ]$ as $t \to \infty$ is obtained by computing the behavior of $\LL(s,n)$ as $s \to 0$, together with Karamata's Tauberian theorem. By the second moment integrability assumption on $H^a(i,j,dt)$, we note that:
\begin{align*}
& m^a(s,i,j) = \int_0^\infty e^{-st} Q^a(i,j,dt)=P^a(i,j) \int_0^\infty e^{-st} H^a(i,j,dt) \\
&\stackrel{s \to 0}{\sim} P^a(i,j)-sP^a(i,j)\int_0^\infty t H^a(i,j,dt)=P^a(i,j)-sP^a(i,j) h^a(i,j).
\end{align*}

Now, assume $P^a(1,1) = P^a(-1,-1)=p_a$. A straightforward but tedious Taylor expansion of $\LL(s,n)$ as $s \to 0$ gives us:
\begin{align*}
& \LL(s,n) \stackrel{s \to 0}{\sim} 1 - \sqrt{\pi} \alpha^a(n) \sqrt{s}.
\end{align*}

The same way, if $P^a(1,1) < P^a(-1,-1)$, a straightforward Taylor expansion of $\LL(s,n)$ as $s \to 0$ gives us:
\begin{align*}
& \LL(s,n) \stackrel{s \to 0}{\sim} 1 - \beta^a(n)s.
\end{align*}
\end{proof}

We are interested in the asymptotic behavior of the law of $\tau$, which is, by independence of the bid/ask queues:

\begin{align*}
& \PP[\tau>t | (q_0^b,q_0^a)=(n_b,n_a) ]=\PP[\sigma_a>t | q_0^a=n_a]\PP[\sigma_b>t | q_0^b=n_b].\\
\end{align*}

We get the following immediate consequence of Proposition \ref{lob1}:

\begin{prop}
\label{lob2}
The conditional law of $\tau$ given $(q_0^b,q_0^a)=(n_b,n_a)$ has a regularly varying tail with:
\begin{itemize}
\item tail exponent 2 if $P^a(1,1) < P^a(-1,-1)$ and $P^b(1,1) < P^b(-1,-1).$ In particular, in this case, $\EE[\tau | (q_0^b,q_0^a)=(n_b,n_a) ]<\infty$.
\item tail exponent 1 if $P^a(1,1) = P^a(-1,-1)$ and $P^b(1,1) = P^b(-1,-1)$. In particular, in this case, $\EE[\tau | (q_0^b,q_0^a)=(n_b,n_a) ]= \infty$ whenever $n_b,n_a \geq 1$.
\item tail exponent 3/2 otherwise. In particular, in this case, $\EE[\tau | (q_0^b,q_0^a)=(n_b,n_a) ]<\infty$.
\end{itemize}

More precisely, we get: if $P^a(1,1) = P^a(-1,-1)$ and $P^b(1,1) = P^b(-1,-1)$:
\begin{align*}
& \PP[\tau>t | (q_0^b,q_0^a)=(n_b,n_a) ] \stackrel{t \to \infty}{\sim} \frac{\alpha^a(n_a) \alpha^b(n_b)}{t}
\end{align*}

if $P^a(1,1) < P^a(-1,-1)$ and $P^b(1,1) < P^b(-1,-1)$:
\begin{align*}
& \PP[\tau>t | (q_0^b,q_0^a)=(n_b,n_a) ] \stackrel{t \to \infty}{\sim} \frac{\beta^a(n_a) \beta^b(n_b)}{t^2}
\end{align*}

if $P^a(1,1) = P^a(-1,-1)$ and $P^b(1,1) < P^b(-1,-1)$:
\begin{align*}
& \PP[\tau>t | (q_0^b,q_0^a)=(n_b,n_a) ] \stackrel{t \to \infty}{\sim} \frac{\alpha^a(n_a) \beta^b(n_b) }{t^{3/2}}
\end{align*}

if $P^a(1,1) < P^a(-1,-1)$ and $P^b(1,1) = P^b(-1,-1)$:
\begin{align*}
& \PP[\tau>t | (q_0^b,q_0^a)=(n_b,n_a) ] \stackrel{t \to \infty}{\sim} \frac{\beta^a(n_a) \alpha^b(n_b) }{t^{3/2}}
\end{align*}

\end{prop}

\begin{proof}  
Immediate using proposition \ref{lob1}.\\
\end{proof}

It will be needed to get the full law of $\tau$, which is, by independence of the bid/ask queues:

\begin{align*}
& \PP[\tau>t | (q_0^b,q_0^a)=(n_b,n_a) ]=\PP[\sigma_a>t | q_0^a=n_a]\PP[\sigma_b>t | q_0^b=n_b].\\
\end{align*}

We have computed explicitely the Laplace transforms of $\sigma_a$ and $\sigma_b$ (cf. the proof of Proposition \ref{lob1} above). There are  two possibilities: either it is possible to invert those Laplace transforms so that we can compute $\PP[\sigma_a>t | q_0^a=n_a]$ and $\PP[\sigma_b>t | q_0^b=n_b]$ in closed form and thus $\PP[\tau>t | (q_0^b,q_0^a)=(n_b,n_a) ]$ in closed form as in \cite{CL}. If not, we will have to resort to a numerical procedure to invert the characteristic functions of $\sigma_a$ and $\sigma_b$. Below we give the characteristic functions of $\sigma_a$ and $\sigma_b$: \\

 \begin{prop}
 \label{lob4}
Let $\phi^a(t,n):=\EE[e^{it \sigma_a} |q_0^a=n]$ ($t \in \RR$) the characteristic function of $\sigma_a$ conditionally on $q_0^a=n \geq 1$. We have:\\

if $m^a(-it,1,1) \neq 0$:
\begin{align*}
& \phi^a(t,n)=\left(c^a(-it)v^a_0(1)+d^a(-it)v^a_0(-1) \right)\lambda^a(-it)^{n-1}, \\
& c^a(z)=\frac{m^a(z,1,-1)}{1-\lambda^a(z)m^a(z,1,1) },\\
& d^a(z)=\frac{m^a(z,-1,1)c^a(z)+\Delta^a(z)}{m^a(z,1,1)},\\
& \Delta^a(z):=m^a(z,1,1)m^a(z,-1,-1)-m^a(z,-1,1)m^a(z,1,-1),\\
&\lambda^a(z):=\frac{1+\Delta^a(z)-\sqrt{(1+\Delta^a(z))^2-4m^a(z,1,1)m^a(z,-1,-1)}}{2m^a(z,1,1)}.
\end{align*}

and if $m^a(-it,1,1) = 0$:
\begin{align*}
& \phi^a(t,n)=(m^a(-it,1,-1) v^a_0(1)+\widetilde{\lambda^a}(-it) v^a_0(-1))\widetilde{\lambda^a}(-it)^{n-1},\\
&\widetilde{\lambda^a}(z):=\frac{m^a(z,-1,-1)}{1-m^a(z,1,-1)m^a(z,-1,1)}.
\end{align*}

The coefficient $\Delta^a(z)$ can be seen as a coupling coefficient and is equal to 0 when the random variable $(V^a_{k},T^a_k)$ doesn't depend on the previous state $V^a_{k-1}$, for example in the context of \cite{CL}.\\

The characteristic function $\phi^b(t,n):=\EE[e^{it \sigma_b} |q_0^b=n]$ has the same expression, with indexes $^a$ replaced by $^b$. 

\end{prop}

\begin{proof}
Similarly to the proof of Proposition \ref{lob1}, we obtain (using the same notations but denoting this time $a_n:=\LL(-it,n,1)$, $b_n:=\LL(-it,n,-1)$):
\begin{align*}
& a_{n+1}=m^a(-it,1,1) a_{n+2}+m^a(-it,1,-1) b_{n}, \hspace{5mm} n \geq 0\\
& b_{n+1}=m^a(-it,-1,1) a_{n+2}+m^a(-it,-1,-1) b_{n}\\
& a_0=b_0=1.
\end{align*}
If $m^a(-it,1,1) = 0$, we can solve explicitly the above system to get the desired result. If $m^a(-it,1,1) \neq 0$, we get as in the proof of Prop \ref{lob1} that both $a_n$ and $b_n$ solve the same following recurrence equation (but for different initial conditions):
\begin{align*}
& m^a(-it,1,1) u_{n+2}-(1+\Delta^a(-it))u_{n+1}+m^a(-it,-1,-1) u_{n}, \hspace{5mm} n \geq 1.\\
\end{align*}
Because $|m^a(-it,j,-1)+m^a(-it,j,1)|=|M^a(-it,j)|=\left| \int_0^\infty e^{its} F^a(j,ds) \right| \leq 1$, tedious computations give us that $|\lambda_+^a(-it)| >1$ whenever $t \neq 0$, where:
\begin{align*}
& \lambda_+^a(z):=\frac{1+\Delta^a(z)+\sqrt{(1+\Delta^a(z))^2-4m^a(z,1,1)m^a(z,-1,-1)}}{2m^a(z,1,1)}.
\end{align*}
Since both $|a_n|, |b_n| \leq 1$ for all $n$, it must be that:
\begin{align*}
& a_n=a_1 \lambda^a(-it)^{n-1}
& b_n=b_1 \lambda^a(-it)^{n-1},\\
\end{align*}
with $a_1,b_1$ being given by the recurrence equations on $a_n$, $b_n$:
\begin{align*}
& a_1=\frac{m^a(-it,1,-1)}{1-\lambda^a(-it)m^a(-it,1,1) }
& b_1=\frac{m^a(-it,-1,1)a_1+\Delta^a(-it)}{m^a(-it,1,1)}.
\end{align*}
Finally we conclude by observing that:
\begin{align*}
& \phi^a(t,n)=a_n v^a_0(1)+b_n v^a_0(-1).
\end{align*}
\end{proof}

\subsection{Probability of Price Increase}

Starting from an initial configuration of the bid and ask queues, $(q_0^b,q_0^a)=(n_b,n_a)$, the probability that the next price change is a price increase will be denoted $p_1^{up}(n_b,n_a)$. This quantity is equal to the probability that $\sigma_a$ is less than $\sigma_b$:
\begin{align*}
& p_1^{up}(n_b,n_a)=\PP[\sigma_a<\sigma_b | q_0^b=n_b, q_0^a=n_a].
\end{align*}

Since we know the characteristic functions of $\sigma_a$, $\sigma_b$ (cf. Proposition \ref{lob4}), we can compute their individual laws up to the use of a numerical procedure. Since $\sigma_a$ and $\sigma_b$ are independent, the law of $\sigma_b-\sigma_a$ can be computed using the individual laws of $\sigma_a$, $\sigma_b$, and therefore $p_1^{up}(n_b,n_a)$ can be computed up to the use of numerical procedures to 1) invert the characteristic function and 2) compute an indefinite integral. Indeed, denoting $f_{n_a,a}$ the p.d.f of $\sigma_a$ conditionally on $q_0^a=n_a$, and $F_{n_b,b}$ the c.d.f. of $\sigma_b$ conditionally on $q_0^b=n_b$, we have:
\begin{align*}
& p_1^{up}(n_b,n_a)=\PP[\sigma_a<\sigma_b | q_0^b=n_b, q_0^a=n_a]=\int_0^\infty f_{n_a,a}(t)(1-F_{n_b,b}(t))dt,\\
\end{align*}

where $F_{n_b,b}$ and $f_{n_a,a}$ are obtained by the following inversion formulas:
\begin{align*}
& f_{n_a,a}(t)=\frac{1}{2 \pi} \int_\RR e^{-itx} \phi^a(x,n_a)dx,\\
& F_{n_b,b}(t)=\frac{1}{2}-\frac{1}{\pi} \int_0^\infty \frac{1}{x} Im\{e^{-itx} \phi^b(x,n_b)\}dx.
\end{align*}

\vspace{3mm}

\subsection{The stock price seen as a functional of a Markov renewal process}

As mentioned earlier, we can write the stock price $s_t$ as:
\begin{align*}
& s_t=\sum_{k=1}^{N_t} X_k,
\end{align*}
where $\{X_n\}_{n \geq 0}$ are the consecutive price increments taking value $\pm \delta$, $\{\tau_n\}_{n \geq 0}$ are the consecutive durations between price changes and $\{T_n\}_{n \geq 0}$ the consecutive times at which the price changes.\\

In this context, the distribution of the random variable $\tau_{n+1}$ will depend on the initial configuration of the bid and ask queues at the beginning $T_n$ of the period $[T_n,T_{n+1})$, which itself depends on the nature of the previous price change $X_n$: if the previous price change is a price decrease, the initial configuration will be drawn from the distribution $\tf$, and if it is an increase, the initial configuration will be drawn from the distribution $f$. Because for each $n$ the random variable $(X_n, \tau_n)$ only depends on the previous increment $X_{n-1}$, it can be seen that the process $(X_n, \tau_n)_{n \geq 0}$ is a Markov renewal process (\cite{LO}, \cite{VS}), and the stock price can therefore be seen as a functional of this Markov renewal process. We obtain the following result.\\

\begin{prop}
\label{lob5}
 The process $(X_n, \tau_n)_{n \geq 0}$ is a Markov renewal process. The law of the process $\{\tau_n\}_{n \geq 0}$ is given by:
\begin{align*}
& F(\delta, t):=\PP[\tau_{n+1} \leq t | X_{n}=\delta]=\sum_{p=1}^\infty \sum_{n=1}^\infty f(n,p) \PP[\tau \leq t | (q_0^b,q_0^a)=(n,p)],\\
& F(-\delta, t):=\PP[\tau_{n+1} \leq t | X_{n}=-\delta]=\sum_{p=1}^\infty \sum_{n=1}^\infty \tf(n,p) \PP[\tau \leq t | (q_0^b,q_0^a)=(n,p)].
\end{align*}

The Markov chain $\{X_n\}_{n \geq 0}$ is characterized by the following transition probabilities:
 \begin{align*}
& p_{cont}:=\PP[X_{n+1}=\delta| X_{n}=\delta]=\sum_{i=1}^\infty \sum_{j=1}^\infty p_1^{up}(i,j) f(i,j).\\
& p'_{cont}:=\PP[X_{n+1}=-\delta| X_{n}=-\delta]=\sum_{i=1}^\infty \sum_{j=1}^\infty (1-p_1^{up}(i,j)) \tf(i,j).
\end{align*}
The generator of this Markov chain is thus (we assimilate the state 1 to the value $\delta$ and the state 2 to the value $-\delta$):
 \[P:= \left(
\begin{array}{ccc}
p_{cont} & 1-p_{cont}  \\
1-p'_{cont} & p'_{cont}
\end{array} \right) \]\\

Let $p_n^{up}(b,a):=\PP[X_{n}=\delta| q_0^b=b,q_0^a=a]$. We can compute this quantity explicitly:
\begin{align*}
& p_n^{up}(b,a)= \pi^*+(p_{cont}+p'_{cont}-1)^{n-1} \left( p_1^{up}(b,a)-\pi^*\right),\\
& \pi^*:=\pi^*(\delta):=\frac{p'_{cont}-1}{p_{cont}+p'_{cont}-2},
\end{align*}

where $\pi^*$ is the stationary distribution of the Markov chain $\{X_n\}$:
 \begin{align*}
& \pi^*=\lim_{n \to \infty} \PP[X_{n}=\delta| X_1].\\
\end{align*}
Further:
\begin{align*}
& \EE[X_{n}| q_0^b=b,q_0^a=a]=\delta(2p_n^{up}(b,a)-1),\\
\end{align*}

and the (conditional) covariance between two consecutive price moves:
\begin{align*}
& cov[X_{n+1}, X_n| q_0^b=b,q_0^a=a]=4\delta^2 p_n^{up}(b,a) (1-p_n^{up}(b,a))(p_{cont}+p'_{cont}-1).\\
\end{align*}
\end{prop}

\begin{rk}
\label{lob6}
In particular, if $p_{cont}=p'_{cont}$, then $\pi^*=1/2$ and we retrieve the results of  \cite{CL}. We also note that the sign of the (conditional) covariance between two consecutive price moves does not depend on the initial configuration of the bid and ask queues and is given by the sign of $p_{cont}+p'_{cont}-1$. We also note that the quantities $p_{cont}$, $p'_{cont}$ can be computed up to the knowledge of the quantities $p_1^{up}(n_b,n_a)$ which computation was discussed in the previous section. The quantities $F( \pm \delta, t)$ can be computed up to the knowledge of the law of $\tau$, which is known up to the use of a numerical procedure to invert the characteristic functions of $\sigma_a$ and $\sigma_b$,  together with the results of Proposition \ref{lob4}.\\
\end{rk}

\begin{proof}

The results follow from elementary calculations in a similar way to what is done in \cite{CL}. Indeed, we have:
 \[
 \left(
\begin{array}{ccc}
p_n^{up}(b,a) & 1-p_n^{up}(b,a) 
\end{array} \right)
 = 
  \left(
\begin{array}{ccc}
p_1^{up}(b,a) & 1-p_1^{up}(b,a) 
\end{array} \right)
 \left(
\begin{array}{ccc}
p_{cont} & 1-p_{cont}  \\
1-p'_{cont} & p'_{cont}
\end{array} \right)^{n-1}
 \]\\

We also have:
 \[
\left(
\begin{array}{ccc}
p_{cont} & 1-p_{cont}  \\
1-p'_{cont} & p'_{cont}
\end{array} \right)
=S 
\left(
\begin{array}{ccc}
1 & 0  \\
0 & p_{cont}+p'_{cont}-1
\end{array} \right) S^{-1}
 \]\\

with:
\[
S=\left(
\begin{array}{ccc}
1 & -\frac{1-p_{cont}}{1-p'_{cont}}  \\
1 & 1
\end{array} \right)
 \]\\
\end{proof}

\section{Diffusion Limit of the Price Process}

In \cite{CL} it is assumed that $f(i,j)=\tf(i,j)=f(j,i)$ in order to make the price increments $X_n$ independent and identically distributed. In fact, this assumption can be entirely relaxed. Indeed, as we mentioned above, $(X_n, \tau_n)_{n \geq 0}$ is in fact a Markov renewal process and therefore we can use the related theory to compute the diffusion limit of the price process. The results of this section generalize the results of Section 4 in \cite{CL}.\\

\subsection{Balanced Order Flow case: $P^a(1,1) = P^a(-1,-1)$ and $P^b(1,1) = P^b(-1,-1)$}

Throughout this section we make the assumption:\\

{\bf (A4)} Using the notations of Proposition \ref{lob1}, the following holds:
\begin{align*}
& \sum_{n=1}^\infty \sum_{p=1}^\infty \alpha^b(n)\alpha^a(p)  f(n,p)<\infty,
& \sum_{n=1}^\infty \sum_{p=1}^\infty \alpha^b(n) \alpha^a(p)  \tf(n,p)<\infty.\\
\end{align*}

Using Proposition \ref{lob2}, we obtain the following result generalizing lemma 1 in \cite{CL}:

\begin{lem}
Under assumption {\bf (A4)}, the following weak convergence holds as $n \to \infty$:
\label{lob7}
\begin{align*}
& \frac{1}{n \log(n)} \sum_{k=1}^n \tau_k \Rightarrow \tau^*:= \sum_{n=1}^\infty \sum_{p=1}^\infty \alpha^b(n) \alpha^a(p)  f^*(n,p),\\
& \mbox{where }f^*(n,p):=\pi^*f(n,p)+(1-\pi^*)\tf(n,p).\\
\end{align*}
\end{lem}

\begin{proof}
We have:
\begin{align*}
& \frac{1}{n \log(n)}\sum_{k=1}^n \tau_k=\sum_{i \in \{-\delta,\delta\}} \frac{N_i(n)}{n} \frac{\log(N_i(n))}{\log(n)} \frac{1}{N_i(n) \log(N_i(n))} \sum_{k=1}^{N_i(n)} \tau_{p(k,i)},
\end{align*}\\

where for $i \in \{-\delta,\delta\}$, $N_{i}(n)$ represents the number of times that $X_{k-1}=i$ for $1 \leq k \leq n$; and $\{p(k,i):k \geq 1\}$ the successive indexes for which $X_{k-1}=i$. By the standard theory of Markov Chains, we have for $i \in \{-\delta,\delta\}$:
\begin{align*}
& \frac{N_i(n)}{n} \stackrel{a.e.}{\to} \pi^*(i),
\end{align*}\\
and therefore we have $\frac{\log(N_i(n))}{\log(n)} \stackrel{a.e.}{\to} 1$. We recall that $\pi^*(\delta):=\pi^*$, and $\pi^*(-\delta)=1-\pi^*$. For fixed $i \in \{-\delta,\delta\}$, the random variables $\{\tau_{p(k,i)}: k \geq 1\}$ are i.i.d. with distribution $F(i, \cdot)$, and with tail index equal to 1 (by Proposition \ref{lob2}). Using \cite{CL} (Lemma 1) together with Proposition \ref{lob2}, we get that:
\begin{align*}
& \frac{1}{n \log(n)} \sum_{k=1}^{n} \tau_{p(k,\delta)} \Rightarrow \sum_{n=1}^\infty \sum_{p=1}^\infty \alpha^b(n) \alpha^a(p) f(n,p),\\
&\frac{1}{n \log(n)} \sum_{k=1}^{n} \tau_{p(k,-\delta)} \Rightarrow \sum_{n=1}^\infty \sum_{p=1}^\infty \alpha^b(n) \alpha^a(p) \tf(n,p).
\end{align*}\\

The latter convergence holds in probability and we finally have:
\begin{align*}
& \frac{1}{n \log(n)} \sum_{k=1}^n \tau_k \stackrel{P}{\to} \pi^* \sum_{n=1}^\infty \sum_{p=1}^\infty \alpha^b(n) \alpha^a(p) f(n,p)+(1-\pi^*)\sum_{n=1}^\infty \sum_{p=1}^\infty \alpha^b(n) \alpha^a(p) \tf(n,p).
\end{align*}
\end{proof}
\vspace{3mm}

Let $s^*:=\delta(2 \pi^*-1)$. Using the previous lemma \ref{lob7}, we obtain the following diffusion limit for the renormalized price process $s_{tn \log(n)}$:\\

\begin{prop}
\label{lob8}
Under assumption {\bf (A4)}, the renormalized price process $s_{tn \log(n)}$ satisfies the following weak convergence in the Skorokhod topology (\cite{S}):
\begin{align*}
& \left(\frac{s_{tn \log(n)}}{n}, t \geq 0 \right) \stackrel{n \to \infty}{\Rightarrow}\left(\frac{s^*t}{\tau^*}, t \geq 0 \right),\\
& \left(\frac{s_{tn \log(n)}-N_{tn \log(n)}s^*}{\sqrt{n}}, t \geq 0 \right) \stackrel{n \to \infty}{\Rightarrow} \frac{\sigma}{\sqrt{\tau^*}} W,\\
\end{align*}
where $W$ is a standard Brownian motion and $\sigma$ is given by:
 \begin{align*}
&  \sigma^2=4 \delta^2 \left( \frac{1-p'_{cont}+\pi^*(p'_{cont}-p_{cont})}{(p_{cont}+p'_{cont}-2)^2}-\pi^*(1-\pi^*)\right).\\
\end{align*}
\end{prop}

\begin{rk}
\label{lob9}
If $p'_{cont}=p_{cont}=\pi^*=\frac{1}{2}$ as in \cite{CL}, we find $s^*=0$ and $\sigma= \delta$ as in \cite{CL}. If $p'_{cont}=p_{cont}=p$, we have $\pi^*=\frac{1}{2}$, $s^*=0$ and:
 \begin{align*}
&  \sigma^2=\delta^2 \frac{p}{1-p}.\\
\end{align*}
\end{rk}

\begin{proof}
Because $m(\pm \delta):=\EE[\tau_n | X_{n-1}=\pm \delta]=+\infty$ by Proposition \ref{lob2}, we cannot directly apply the well-known invariance principle results for semi-Markov processes. Denote for $t \in \RR_+$:
\begin{align*}
& R_n:=\sum_{k=1}^{n} (X_k-s^*),
& U_n(t):=n^{-1/2}\left[(1-\lambda _{n,t}) R_{\lfloor n t \rfloor} +\lambda _{n,t} R_{\lfloor n t \rfloor+1} \right],
\end{align*}
where $\lambda _{n,t}:=n t-\lfloor n t \rfloor$. We can show, following a martingale method similar to \cite{VS} (section 3), that we have the following weak convergence in the Skorokhod topology:
\begin{align*}
(U_n(t), t \geq 0) \stackrel{n \to \infty}{\Rightarrow} \sigma W,
\end{align*}
 where $W$ is a standard Brownian motion, and $\sigma$ is given by:
 \begin{align*}
& \sigma^2 = \sum_{i \in \{-\delta,\delta\}} \pi^*(i) v(i),
\end{align*}

where for $i \in \{-\delta,\delta\}$:
 \begin{align*}
& v(i):=b(i)^2+p(i)(g(-i)-g(i))^2-2b(i)p(i)(g(-i)-g(i)),\\
&b(i):=i-s^*,\\
& p(\delta):=1-p_{cont}, \hspace{3mm} p(-\delta):=1-p'_{cont},\\
\end{align*}

and (the vector) $g$ is given by:
 \begin{align*}
&g=(P+\Pi^*-I)^{-1}b,
\end{align*}
 where $\Pi^*$ is the matrix with rows equal to $(\pi^* \hspace{3mm}1-\pi^*)$. After completing the calculations we get:
 \begin{align*}
&  \sigma^2=4 \delta^2 \left( \frac{1-p'_{cont}+\pi^*(p'_{cont}-p_{cont})}{(p_{cont}+p'_{cont}-2)^2}-\pi^*(1-\pi^*)\right).
\end{align*}

For the sake of exhaustivity we also give the explicit expression for $g$: 
 \begin{align*}
&  g(\delta)=\delta\frac{p'_{cont}-p_{cont}+2(1-\pi^*)}{p_{cont}+p'_{cont}-2}-s^*,\\
& g(-\delta)=\delta\frac{p'_{cont}-p_{cont}-2\pi^*}{p_{cont}+p'_{cont}-2}-s^*.
\end{align*}

Indeed, to show the above convergence of $U_n$, we observe that we can write $R_n$ as the sum of a $\F_n-$martingale $M_n$ and a bounded process:
\begin{align*}
& R_n=M_n+\underbrace{g(X_n)-g(X_0)+X_n-X_0}_{unif. bounded},
& M_n:=\sum_{k=1}^n b(X_{k-1})-g(X_k)+g(X_{k-1}),
\end{align*}
where $\F_n := \sigma(\tau_k,X_k: k \leq n)$ and $X_0:=0$. The process $M_n$ is a martingale because $g$ is the unique solution of the following Poisson equation, since $\Pi^*b=0$:
\begin{align*}
& [P-I]g=b.
\end{align*}

The rest of the proof for the convergence of $U_n$ follows exactly \cite{VS} (section 3). \\

We proved earlier (lemma \ref{lob7}) that:
 \begin{align*}
& \frac{T_n}{n \log(n)} \Rightarrow \tau^*,\\
\end{align*}
where $T_n:=\sum_{k=1}^n \tau_k$. Since the Markov renewal process $(X_n,\tau_n)_{n \geq 0}$ is regular (because the state space is finite), we get $N_t \to \infty$ a.s. and therefore:
 \begin{align*}
& \frac{T_{N_t}}{N_t \log(N_t)} \Rightarrow \tau^*.\\
\end{align*}

Observing that $T_{N_t} \leq t \leq T_{N_t+1}$ a.s., we get:
\begin{align*}
& \frac{T_{N_t}}{N_t \log N_t} \leq \frac{t}{N_t \log N_t} \leq \frac{(N_t+1) \log (N_t+1)}{N_t \log N_t} \frac{T_{N_t+1}}{(N_t+1) \log (N_t+1)},
 \end{align*}\\

and therefore:
 \begin{align*}
& \frac{t}{N_t \log(N_t)} \Rightarrow \tau^*.\\
\end{align*}

Let $t_n:=tn\log(n)$. We would like to show as in \cite{CL}, equation $(17)$ that:
 \begin{align*}
& N_{t_n} \stackrel{P}{\sim} \frac{nt}{\tau^*}.\\
\end{align*}

We have denoted by $A_n \stackrel{P}{\sim} B_n$ iff $P-\lim \frac{A_n}{B_n}=1$. We denote as in \cite{CL} $\rho:(1,\infty) \to (1, \infty)$ to be the inverse function of $t\log(t)$, and we note that $\rho(t) \stackrel{t \to \infty}{\sim} \frac{t}{\log(t)}$. The first equivalence in \cite{CL}, equation $(17)$: $N_{t_n} \stackrel{P}{\sim} \rho \left( \frac{t_n}{\tau^*}\right)$ is not obvious. Indeed, we have $N_{t_n} \log(N_{t_n}) \stackrel{P}{\sim} \frac{t_n}{\tau^*}$, and we would like to conclude that $N_{t_n}=\rho(N_{t_n} \log(N_{t_n})) \stackrel{P}{\sim} \rho \left( \frac{t_n}{\tau^*}\right)$. The latter implication is not true for every function $\rho$, in particular if $\rho$ was exponential. Nevertheless, in our case, it is true because $\rho(t) \stackrel{t \to \infty}{\sim} \frac{t}{\log(t)}$, and therefore for any functions $f,g$ going to $+\infty$ as $t \to \infty$:
 \begin{align*}
& \frac{\rho(f(t))}{\rho(g(t))} \stackrel{t \to \infty}{\sim}  \frac{f(t)}{g(t)} \frac{\log(g(t))}{\log(f(t))}.
\end{align*}

Therefore we see that if $f(t) \stackrel{t \to \infty}{\sim} g(t)$, then by property of the logarithm $\log(f(t)) \stackrel{t \to \infty}{\sim} \log(g(t))$ and therefore $\rho(f(t)) \stackrel{t \to \infty}{\sim} \rho(g(t))$. This allows us to conclude as in \cite{CL} that:
 \begin{align*}
& \frac{N_{t_n}}{n} \stackrel{P}{\sim} \frac{t}{\tau^*}.\\
\end{align*}

Therefore, we can make a change of time as in \cite{VS}, Corollary 3.19 (see also \cite{B}, section 14), and denoting $\alpha_n(t):=\frac{N_{t_n}}{n}$, we obtain the following weak convergence in the Skorohod topology:
\begin{align*}
(U_n(\alpha_n(t)), t \geq 0) \Rightarrow (\sigma W_{\frac{t}{\tau^*}}, t \geq 0), 
\end{align*}

that is to say:
\begin{align*}
\left(\frac{s_{tn \log(n)}-N_{tn \log(n)}s^*}{\sqrt{n}}, t \geq 0 \right) \Rightarrow \frac{\sigma}{\sqrt{\tau^*}} W.
\end{align*}

The law of large numbers result comes from the fact that $\frac{N_{t_n}}{n} \stackrel{P}{\sim} \frac{t}{\tau^*}$, together with the following fact (strong law of large numbers for Markov chains):
\begin{align*}
& \frac{1}{n} \sum_{k=1}^n X_k \to s^* \hspace{3mm} a.e.
\end{align*}
\end{proof}

\subsection{Other cases: either $P^a(1,1) < P^a(-1,-1)$ or $P^b(1,1)<P^b(-1,-1)$} 

In this case, we know by Proposition \ref{lob2} that the conditional expectations $\EE[\tau_k| q^b_0=n_b, q^a_0=n_a ]$ are finite. Denoting the conditional expectations $m(\pm \delta):=\EE[\tau_k | X_{k-1}=\pm \delta]$, we have:
\begin{align*}
& m(\delta)=\sum_{p=1}^\infty \sum_{n=1}^\infty \EE[\tau_k| q^b_0=n, q^a_0=p ] f(n,p),\\
&m(-\delta)=\sum_{p=1}^\infty \sum_{n=1}^\infty \EE[\tau_k| q^b_0=n, q^a_0=p ] \tf (n,p).
\end{align*}

Throughout this section we will need the following assumption:\\

{\bf (A5)} Using the previous notations, the following holds:
\begin{align*}
& m(\pm \delta)<\infty. \\
\end{align*}

For example, the above assumption is satisfied if the support of the distributions $f$ and $\tf$ is compact, which is the case in practice. We obtain the following diffusion limit result as a classical consequence of invariance principle results for semi-Markov processes (see e.g. \cite{VS}, section 3):\\
\begin{prop}
\label{lob10}
Under assumption {\bf (A5)}, the renormalized price process $s_{nt}$ satisfies the following convergence in the Skorokhod topology:
\begin{align*}
& \left(\frac{s_{n t}}{n}, t \geq 0 \right) \stackrel{n \to \infty}{\to}\left(\frac{s^*t}{m_\tau}, t \geq 0 \right)  \hspace{3mm} \mbox{a.e.},\\
& \left(\frac{s_{n t}- N_{nt}s^*}{\sqrt{n}}, t \geq 0 \right) \stackrel{n \to \infty}{\Rightarrow} \frac{\sigma}{\sqrt{m_\tau}} W,\\
\end{align*}
where $W$ is a standard Brownian motion, $\sigma$ is given in Proposition \ref{lob8} and:
\begin{align*}
& m_\tau:=\sum_{i \in \{-\delta,\delta\}} \pi^*(i) m(i)=\pi^*m(\delta)+(1-\pi^*)m(-\delta).
\end{align*}
\end{prop}


\begin{proof}
This is an immediate consequence of strong law of large numbers and invariance principle results for Markov renewal processes satisfying $m(\pm \delta)<\infty$ (see e.g. \cite{VS} section 3). In the previous article \cite{VS}, the proof of the invariance principle is carried on using a martingale method similar to the one of the proof of proposition \ref{lob8}.
\end{proof}

\section{Numerical Results}

In this section, we present calibration results which illustrate and justify our approach.\\

In \cite{CL}, it is assumed that the queue changes $V^b_k,V^a_k$ do not depend on their previous values $V^b_{k-1},V^a_{k-1}$. Empirically, it is found that $\PP[V^b_k=1] \approx \PP[V^b_k=-1]\approx 1/2$ (and similarly for the ask side). Here, we challenge this assumption by estimating and comparing the probabilities $P(-1,1)$ Vs. $P(1,1)$ on the one side and $P(-1,-1)$ Vs. $P(1,-1)$ on the other side to check whether or not they are approximately equal to each other, for both the ask and the bid. We also give - for both the bid and ask - the estimated probabilities $\PP[V_k=1]$, $\PP[V_k=-1]$ that we call respectively $P(1)$, $P(-1)$, to check whether or not they are approximately equal to 1/2 as in \cite{CL}.\\

The results below correspond to the 5 stocks Amazon, Apple, Google, Intel, Microsoft on June $21^{st}$ 2012 \footnote{The data was taken from the webpage \textit{https://lobster.wiwi.hu-berlin.de/info/DataSamples.php}}. The probabilities are estimated using the strong law of large numbers. We also give for indicative purposes the average time between order arrivals (in milliseconds (ms)) as well as the average number of stocks per order.\\

\begin{center}
\begin{tabular}{|c|c|c|c|c|c|c|c|c|c|c|}
\hline
\multirow{2}{1.5cm}{} &
\multicolumn{2}{c|}{Amazon}&
\multicolumn{2}{c|}{Apple}&
\multicolumn{2}{c|}{Google}&
\multicolumn{2}{c|}{Intel}&
\multicolumn{2}{c|}{Microsoft}\\
\cline{2-11}
& Bid & Ask &Bid&Ask&Bid&Ask&Bid&Ask&Bid&Ask\\\hline
Avg time btw. orders (ms) & 910 & 873 & 464 & 425 & 1123 & 1126 & 116 & 133 &130 &113\\
Avg nb. of stocks per order & 100 & 82 & 90 & 82 & 84 & 71 & 502 & 463 & 587 & 565 \\\hline
\end{tabular}\\
$\left. \right.$\\
\scriptsize
\textit{Average time between orders (ms) \& Average number of stocks per order. June $21^{st}$ 2012.}\\
\end{center}

\begin{center}
\begin{tabular}{|c|c|c|c|c|c|c|c|c|c|c|}
\hline
\multirow{2}{1.5cm}{} &
\multicolumn{2}{c|}{Amazon}&
\multicolumn{2}{c|}{Apple}&
\multicolumn{2}{c|}{Google}&
\multicolumn{2}{c|}{Intel}&
\multicolumn{2}{c|}{Microsoft}\\
\cline{2-11}
& Bid & Ask &Bid&Ask&Bid&Ask&Bid&Ask&Bid&Ask\\\hline
$P(1,1)$ & 0.48 & 0.57  & 0.50 & 0.55 & 0.48 & 0.53 & 0.55 & 0.61& 0.63 & 0.60 \\
$P(-1,1)$ & 0.46 & 0.42 & 0.40 & 0.42 & 0.46 & 0.49 & 0.44 & 0.40& 0.36& 0.41\\\hline
$P(-1,-1)$ & 0.54 & 0.58 & 0.60& 0.58 & 0.54 & 0.51& 0.56 & 0.60 & 0.64& 0.59\\
$P(1,-1)$ & 0.52 & 0.43 & 0.50& 0.45& 0.52& 0.47& 0.45 & 0.39& 0.37 & 0.40\\\hline
$P(1)$ & 0.47 & 0.497 & 0.44& 0.48 & 0.47& 0.51& 0.495 & 0.505& 0.49& 0.508\\
$P(-1)$ & 0.53 & 0.503 & 0.56& 0.52& 0.53& 0.49& 0.505 & 0.495& 0.51 & 0.492\\\hline
\end{tabular}\\
$\left. \right.$\\
\scriptsize
\textit{Estimated transition probabilities of the Markov Chains $V^b_k,V^a_k$. June $21^{st}$ 2012.}\\
\end{center}

\textbf{Findings:} First of all, we find as in \cite{CL} that for all stocks, $\PP[V_k=1] \approx \PP[V_k=-1]\approx 1/2$, except maybe in the case of Apple Bid. It is worth mentioning that we always have $P(1) < P(-1)$ except in 3 cases: Google Ask, Intel Ask and Microsoft Ask. Nevertheless, in these cases, $P(1)$ and $P(-1)$ are very close to each other and so they could be considered to fall into the case $P(1)=P(-1)$ of \cite{CL}. These 3 cases also correspond to the only 3 cases where $P(1,1)>P(-1,-1)$, which is contrary to our assumption $P(1,1) \leq P(-1,-1)$. Nevertheless, in these 3 cases, $P(1,1)$ and $P(-1,-1)$ are very close to each other so we can consider them to fall into the case $P(1,1)=P(-1,-1)$.\\

More importantly, we notice that the probabilities $P(-1,1)$, $P(1,1)$ can be  significantly different from each other - and similarly for the probabilities $P(-1,-1)$, $P(1,-1)$ - which justifies the use of a Markov Chain structure for the random variables $\{V^b_k\},\{V^a_k\}$. This phenomenon is particularly visible for example on Microsoft (Bid+Ask), Intel (Bid+Ask), Apple (Bid+Ask) or Amazon Ask. Further, regarding the comparison of $P(1,1)$ and $P(-1,-1)$, it turns out that they are often very smilar, except in the cases Amazon Bid, Apple Bid, Google Bid.\\

The second assumption of \cite{CL} that we would like to challenge is the assumed exponential distribution of the order arrival times $T^a_k,T_k^b$. To this end, on the same data set as used to estimate the transition probabilities $P^a(i,j)$, $P^b(i,j)$, we calibrate the empirical c.d.f.'s $H^a(i,j,\cdot)$, $H^b(i,j,\cdot)$ to the Gamma and Weibull distributions (which are generalizations of the exponential distribution). We recall that the p.d.f.'s of these distributions are given by:
\begin{align*}
& f_{Gamma}(x)=\frac{1}{\Gamma(k) \theta^k} x^{k-1} e^{-\frac{x}{\theta}}1_{x > 0},\\
& f_{Weibull}(x)= \frac{k}{\theta} \left( \frac{x}{\theta} \right)^{k-1} e^{-\left( \frac{x}{\theta} \right)^{k}}1_{x > 0}.
\end{align*}
Here, $k>0$ and $\theta>0$ represent respectively the shape and the scale parameter. The variable $k$ is dimensionless, whereas $\theta$ will be expressed in $ms^{-1}$. We perform a maximum likelihood estimation of the Weibull and Gamma parameters for each one of the empirical distributions $H^a(i,j, \cdot)$, $H^b(i,j,\cdot)$ (together with a 95 \% confidence interval for the parameters). As we can see on the tables below, the shape parameter $k$ is always significantly different than 1 ($\sim 0.1$ to $0.3$), which indicates that the exponential distribution is not rich enough to fit our observations. To illustrate this, we present below the empirical c.d.f. of $H(1,-1)$ in the case of Google Bid, and we see that Gamma and Weibull allow to fit the empirical c.d.f. in a much better way than Exponential.

\scriptsize
\begin{center}
\includegraphics[scale=0.5]{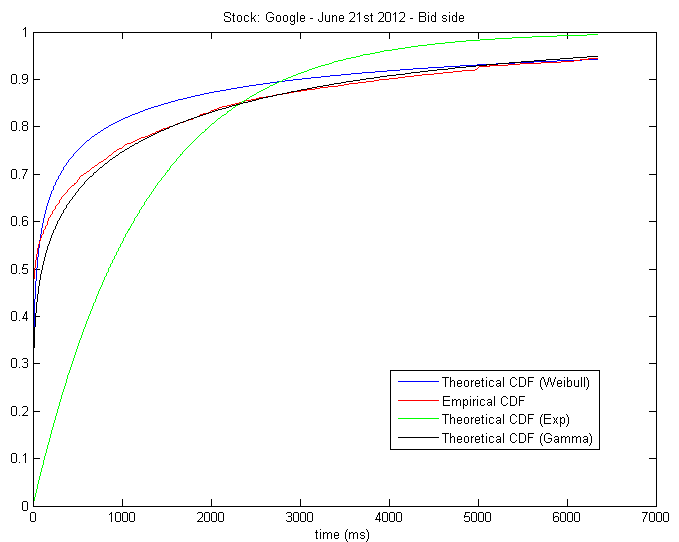}\\
$\left. \right.$\\
\textit{H(1,-1) - Google Bid - June $21^{st}$ 2012.}\\
\end{center}
\normalsize
\vspace{5mm}

We summarize our calibration results in the tables below.  \\

\begin{center}
\begin{tabular}{|c|c|c|c|c|}
\hline
\centering \textbf{Amazon Bid} 
& \small$H(1,1)$ & \small$H(1,-1)$ &\small $H(-1,-1)$& \small$H(-1,1)$ \\\hline
Weibull $\theta$ & 99.1  & 185.5  & 87.7 & 87.0  \\
& \scriptsize (90.2-109.0)  & \scriptsize (171.3-200.8)  & \scriptsize (80.1-96.0) & \scriptsize (78.7-96.1) \\
Weibull $k$ & 0.279  & 0.323  & 0.285 & 0.258 \\
& \scriptsize (0.274-0.285)  & \scriptsize (0.317-0.329)  & \scriptsize (0.280-0.290) & \scriptsize (0.253-0.263)   \\\hline\hline
Gamma $\theta$ & 4927 & 4321  & 4712 & 5965  \\
& \scriptsize (4618-5257)  & \scriptsize (4075-4582)  & \scriptsize (4423-5019) & \scriptsize (5589-6366)  \\
Gamma $k$ & 0.179 & 0.215  & 0.179 & 0.165  \\
& \scriptsize (0.174-0.184)  & \scriptsize (0.209-0.220)  & \scriptsize (0.175-0.184) & \scriptsize (0.161-0.169)  \\\hline
\end{tabular}\\
$\left. \right.$\\
\scriptsize
\textit{Amazon Bid: Fitted Weibull and Gamma parameters. 95 \% confidence intervals in brackets. June $21^{st}$ 2012.}\\
\end{center}

\begin{center}
\begin{tabular}{|c|c|c|c|c|}
\hline
\centering \textbf{Amazon Ask} 
& \small$H(1,1)$ & \small$H(1,-1)$ &\small $H(-1,-1)$& \small$H(-1,1)$ \\\hline
Weibull $\theta$ & 80.8 & 197.8 & 57.9 & 137.0  \\
& \scriptsize (74.4-87.7) & \scriptsize (181.9-215.1) & \scriptsize (52.8-63.4) & \scriptsize (124.2-151.2) \\
Weibull $k$ & 0.274 & 0.324 & 0.279 & 0.276 \\
& \scriptsize (0.269-0.278) & \scriptsize (0.317-0.330) & \scriptsize (0.274-0.285) & \scriptsize (0.270-0.281)  \\\hline\hline
Gamma $\theta$ & 4732 & 4623 & 3845 & 5879  \\
& \scriptsize (4475-5004) & \scriptsize (4345-4919) & \scriptsize (3609-4095) & \scriptsize (5502-6283)  \\
Gamma $k$ & 0.174 & 0.215 & 0.173 & 0.181  \\
&  \scriptsize (0.170-0.178) & \scriptsize (0.209-0.221) & \scriptsize (0.168-0.177) & \scriptsize (0.176-0.186)  \\\hline
\end{tabular}\\
$\left. \right.$\\
\scriptsize
\textit{Amazon Ask: Fitted Weibull and Gamma parameters. 95 \% confidence intervals in brackets. June $21^{st}$ 2012.}\\
\end{center}

\begin{center}
\begin{tabular}{|c|c|c|c|c|}
\hline
\centering \textbf{Apple Bid} 
& \small$H(1,1)$ & \small$H(1,-1)$ &\small $H(-1,-1)$& \small$H(-1,1)$ \\\hline
Weibull $\theta$ & 75.9  & 180.9  & 31.5 & 78.2  \\
& \scriptsize (71.6-80.5)  & \scriptsize (172.6-189.7)  & \scriptsize (29.5-33.6) & \scriptsize (73.4-83.3) \\
Weibull $k$ & 0.317  & 0.400  & 0.271 & 0.300 \\
& \scriptsize (0.313-0.321)  & \scriptsize (0.394-0.405)  & \scriptsize (0.267-0.274) & \scriptsize (0.296-0.304)   \\\hline\hline
Gamma $\theta$ & 2187 & 1860  & 2254 & 2711  \\
& \scriptsize (2094-2284)  & \scriptsize (1787-1935)  & \scriptsize (2157-2355) & \scriptsize (2592-2835)  \\
Gamma $k$ & 0.206 & 0.276  & 0.168 & 0.196  \\
& \scriptsize (0.202-0.210)  & \scriptsize (0.271-0.282)  & \scriptsize (0.165-0.171) & \scriptsize (0.192-0.199)  \\\hline
\end{tabular}\\
$\left. \right.$\\
\scriptsize
\textit{Apple Bid: Fitted Weibull and Gamma parameters. 95 \% confidence intervals in brackets. June $21^{st}$ 2012.}\\
\end{center}

\begin{center}
\begin{tabular}{|c|c|c|c|c|}
\hline
\centering \textbf{Apple Ask} 
& \small$H(1,1)$ & \small$H(1,-1)$ &\small $H(-1,-1)$& \small$H(-1,1)$ \\\hline
Weibull $\theta$ & 46.6 & 152.5 & 27.7 & 95.5  \\
& \scriptsize (44.1-49.2) & \scriptsize (145.5-159.8) & \scriptsize (26.0-29.6) & \scriptsize (90.0-101.5) \\
Weibull $k$ & 0.298 & 0.394 & 0.271 & 0.308 \\
& \scriptsize (0.294-0.301) & \scriptsize (0.388-0.399) & \scriptsize (0.267-0.275) & \scriptsize (0.303-0.312) \\\hline\hline
Gamma $\theta$ & 2019 & 1666 & 1995 & 2740  \\
& \scriptsize (1942-2099) & \scriptsize (1603-1732) & \scriptsize (1907-2087) & \scriptsize (2624-2861)  \\
Gamma $k$ & 0.189 & 0.271 & 0.168 & 0.204  \\
&  \scriptsize (0.186-0.192) & \scriptsize (0.266-0.277) & \scriptsize (0.165-0.171) & \scriptsize (0.200-0.208)  \\\hline
\end{tabular}\\
$\left. \right.$\\
\scriptsize
\textit{Apple Ask: Fitted Weibull and Gamma parameters. 95 \% confidence intervals in brackets. June $21^{st}$ 2012.}\\
\end{center}

\begin{center}
\begin{tabular}{|c|c|c|c|c|}
\hline
\centering \textbf{Google Bid} 
& \small$H(1,1)$ & \small$H(1,-1)$ &\small $H(-1,-1)$& \small$H(-1,1)$ \\\hline
Weibull $\theta$ & 113.9  & 158.5  & 67.9 & 56.8  \\
& \scriptsize (102.8-126.2)  & \scriptsize (143.4-175.3)  & \scriptsize (60.6-76.0) & \scriptsize (50.5-63.8) \\
Weibull $k$ & 0.276  & 0.284  & 0.261 & 0.246 \\
& \scriptsize (0.270-0.282)  & \scriptsize (0.278-0.290)  & \scriptsize (0.255-0.266) & \scriptsize (0.241-0.251)  \\\hline\hline
Gamma $\theta$ & 6720 & 6647  & 6381 & 7025   \\
& \scriptsize (6263-7210)  & \scriptsize (6204-7122)  & \scriptsize (5913-6886) & \scriptsize (6517-7571)  \\
Gamma $k$ & 0.174 & 0.185  & 0.160 & 0.151  \\
& \scriptsize (0.169-0.179)  & \scriptsize (0.180-0.191)  & \scriptsize (0.155-0.165) & \scriptsize (0.147-0.156)  \\\hline
\end{tabular}\\
$\left. \right.$\\
\scriptsize
\textit{Google Bid: Fitted Weibull and Gamma parameters. 95 \% confidence intervals in brackets. June $21^{st}$ 2012.}\\
\end{center}

\begin{center}
\begin{tabular}{|c|c|c|c|c|}
\hline
\centering \textbf{Google Ask} 
& \small$H(1,1)$ & \small$H(1,-1)$ &\small $H(-1,-1)$& \small$H(-1,1)$ \\\hline
Weibull $\theta$ & 196.7 & 271.6 & 38.1 & 57.0  \\
& \scriptsize (180.6-214.2) & \scriptsize (248.5-296.8) & \scriptsize (33.8-43.0) & \scriptsize (51.3-63.3) \\
Weibull $k$ & 0.290 & 0.310 & 0.258 & 0.263 \\
& \scriptsize (0.285-0.295) & \scriptsize (0.303-0.316) & \scriptsize (0.253-0.264) & \scriptsize (0.258-0.268)  \\\hline\hline
Gamma $\theta$ & 6081 & 6571 & 4304 & 4698  \\
& \scriptsize (5734-6450) & \scriptsize (6165-7003) & \scriptsize (3971-4664) & \scriptsize (4380-5040)  \\
Gamma $k$ & 0.195 & 0.209 & 0.156 & 0.164  \\
&  \scriptsize (0.190-0.200) & \scriptsize (0.203-0.215) & \scriptsize (0.151-0.161) & \scriptsize (0.159-0.168) \\\hline
\end{tabular}\\
$\left. \right.$\\
\scriptsize
\textit{Google Ask: Fitted Weibull and Gamma parameters. 95 \% confidence intervals in brackets. June $21^{st}$ 2012.}\\
\end{center}

\begin{center}
\begin{tabular}{|c|c|c|c|c|}
\hline
\centering \textbf{Intel Bid} 
& \small$H(1,1)$ & \small$H(1,-1)$ &\small $H(-1,-1)$& \small$H(-1,1)$ \\\hline
Weibull $\theta$ & 2.76  & 2.56  & 3.33 & 2.01  \\
& \scriptsize (2.66-2.86)  & \scriptsize (2.45-2.67)  & \scriptsize (3.21-3.45) & \scriptsize (1.92-2.10) \\
Weibull $k$ & 0.227  & 0.226  & 0.267 & 0.209 \\
& \scriptsize (0.226-0.229)  & \scriptsize (0.225-0.228)  & \scriptsize (0.265-0.269) & \scriptsize (0.208-0.211)  \\\hline\hline
Gamma $\theta$ & 1016 & 912  & 543 & 1093   \\
& \scriptsize (991-1040)  & \scriptsize (888-937)  & \scriptsize (530-557) & \scriptsize (1063-1124)  \\
Gamma $k$ & 0.129 & 0.130  & 0.151 & 0.120  \\
& \scriptsize (0.128-0.130)  & \scriptsize (0.129-0.131)  & \scriptsize (0.150-0.152) & \scriptsize (0.119-0.121)  \\\hline
\end{tabular}\\
$\left. \right.$\\
\scriptsize
\textit{Intel Bid: Fitted Weibull and Gamma parameters. 95 \% confidence intervals in brackets. June $21^{st}$ 2012.}\\
\end{center}

\begin{center}
\begin{tabular}{|c|c|c|c|c|}
\hline
\centering \textbf{Intel Ask} 
& \small$H(1,1)$ & \small$H(1,-1)$ &\small $H(-1,-1)$& \small$H(-1,1)$ \\\hline
Weibull $\theta$ & 1.33  & 5.46  & 4.63 & 5.15  \\
& \scriptsize (1.28-1.38)  & \scriptsize (5.21-5.73)  & \scriptsize (4.45-4.80) & \scriptsize (4.90-5.41) \\
Weibull $k$ & 0.235  & 0.231  & 0.256 & 0.225 \\
& \scriptsize (0.234-0.237)  & \scriptsize (0.230-0.233)  & \scriptsize (0.254-0.257) & \scriptsize (0.224-0.227)  \\\hline\hline
Gamma $\theta$ & 705 & 1219  & 884 & 1305   \\
& \scriptsize (688-723)  & \scriptsize (1183-1256)  & \scriptsize (862-907) & \scriptsize (1266-1345)  \\
Gamma $k$ & 0.126 & 0.137  & 0.146 & 0.133  \\
& \scriptsize (0.125-0.127)  & \scriptsize (0.136-0.139)  & \scriptsize (0.144-0.147) & \scriptsize (0.132-0.135)  \\\hline
\end{tabular}\\
$\left. \right.$\\
\scriptsize
\textit{Intel Ask: Fitted Weibull and Gamma parameters. 95 \% confidence intervals in brackets. June $21^{st}$ 2012.}\\
\end{center}

\begin{center}
\begin{tabular}{|c|c|c|c|c|}
\hline
\centering \textbf{Microsoft Bid} 
& \small$H(1,1)$ & \small$H(1,-1)$ &\small $H(-1,-1)$& \small$H(-1,1)$ \\\hline
Weibull $\theta$ & 0.79  & 2.98  & 2.68 & 2.64  \\
& \scriptsize (0.76-0.82)  & \scriptsize (2.83-3.13)  & \scriptsize (2.59-2.78) & \scriptsize (2.50-2.78) \\
Weibull $k$ & 0.215  & 0.221  & 0.259 & 0.211 \\
& \scriptsize (0.214-0.217)  & \scriptsize (0.219-0.223)  & \scriptsize (0.257-0.260) & \scriptsize (0.209-0.213)  \\\hline\hline
Gamma $\theta$ & 1012 & 1315  & 664 & 1488   \\
& \scriptsize (987-1039)  & \scriptsize (1274-1358)  & \scriptsize (648-681) & \scriptsize (1440-1537)  \\
Gamma $k$ & 0.112 & 0.125  & 0.142 & 0.120  \\
& \scriptsize (0.111-0.113)  & \scriptsize (0.124-0.127)  & \scriptsize (0.141-0.143) & \scriptsize (0.118-0.121)  \\\hline
\end{tabular}\\
$\left. \right.$\\
\scriptsize
\textit{Microsoft Bid: Fitted Weibull and Gamma parameters. 95 \% confidence intervals in brackets. June $21^{st}$ 2012.}\\
\end{center}

\begin{center}
\begin{tabular}{|c|c|c|c|c|}
\hline
\centering \textbf{Microsoft Ask} 
& \small$H(1,1)$ & \small$H(1,-1)$ &\small $H(-1,-1)$& \small$H(-1,1)$ \\\hline
Weibull $\theta$ & 0.85  & 1.57  & 2.07 & 1.43  \\
& \scriptsize (0.82-0.89)  & \scriptsize (1.50-1.64)  & \scriptsize (2.00-2.15) & \scriptsize (1.36-1.50) \\
Weibull $k$ & 0.218  & 0.223  & 0.259 & 0.210 \\
& \scriptsize (0.217-0.219)  & \scriptsize (0.222-0.225)  & \scriptsize (0.258-0.261) & \scriptsize (0.208-0.211)  \\\hline\hline
Gamma $\theta$ & 1004 & 1081  & 574 & 1138   \\
& \scriptsize (980-1028)  & \scriptsize (1051-1112)  & \scriptsize (560-588) & \scriptsize (1105-1171)  \\
Gamma $k$ & 0.113 & 0.121  & 0.140 & 0.116  \\
& \scriptsize (0.112-0.114)  & \scriptsize (0.120-0.122)  & \scriptsize (0.139-0.141) & \scriptsize (0.115-0.117)  \\\hline
\end{tabular}\\
$\left. \right.$\\
\scriptsize
\textit{Microsoft Ask: Fitted Weibull and Gamma parameters. 95 \% confidence intervals in brackets. June $21^{st}$ 2012.}\\
\end{center}

\section{Conclusion and Future Work}

In this paper, we introduced a semi-Markovian modeling of limit order books in order to match empirical observations. We extended the model of \cite{CL} in the following ways:
\begin{enumerate}[i)]
\item inter-arrival times between book events (limit orders, market orders, order cancellations) are allowed to have an arbitrary distribution.
\item the arrival of a new book event at the bid or the ask and its corresponding inter-arrival time are allowed to depend on the nature of the previous event.
\end{enumerate}

In order to do so, both the bid and ask queues are driven by Markov renewal processes. It results from these chosen dynamics that the price process can be expressed as a functional of another Markov renewal process, which we characterized explicitly. In this context, we obtained probabilistic results such as the duration until the next price change, the probability of price increase and the characterization of the Markov renewal process driving the stock price process (section 3). In section 4, we obtained diffusion limit results for the stock price process generalizing those of \cite{CL}. Finally, we presented in section 5 calibration results on real market data in order to illustrate and justify our approach.


\end{document}